%% file: cutwidth.tex
\documentclass[11pt,letter]{article}
\usepackage{vmargin}
\usepackage{comment}
\usepackage{enumerate}
 \setmarginsrb{1in}{1in}{1in}{1in}{0pt}{0pt}{0pt}{7mm}
\usepackage{tikz}
\usetikzlibrary{snakes}

\usepackage{makeidx}  

\usepackage{amsmath,amstext,amssymb}
\usepackage{xspace}
  \usepackage{amsthm}

  \newtheorem{theorem}{Theorem}
  \newtheorem{lemma}[theorem]{Lemma}
    \newtheorem{corollary}[theorem]{Corollary}
  
  \newtheorem{definition}{Definition}

\begin{document}
%

\newcommand{\tudu}[1]{{\bf{TODO: #1}}}
\newcommand{\transname}[2]{$#1 \slash #2$-{\sc{transversal}}}
\newcommand{\himmersion}{{$\mathbf{H}$\sc{-Immersion}}\xspace}
\newcommand{\dpath}{{\sc{Edge Disjoint Paths}}\xspace}
\newcommand{\vertexpath}{{\sc{Vertex Disjoint Paths}}\xspace}
\newcommand{\fast}{{\sc{Feedback Arc Set}}\xspace}
\newcommand{\olaf}{{\sc{Optimal Linear Arrangement}}\xspace}
\newcommand{\SAT}{{\sc{SAT}}\xspace}
\newcommand{\compass}{\ensuremath{\textrm{NP} \subseteq \textrm{coNP}/\textrm{poly}}}
\newcommand{\N}{\ensuremath{\mathbb{N}}}
\newcommand{\Aa}{\ensuremath{\mathcal{A}}}
\newcommand{\Bb}{\ensuremath{\mathcal{B}}}
\newcommand{\Jj}{\ensuremath{\mathcal{J}}}
\newcommand{\Pp}{\ensuremath{\mathcal{P}}}
\newcommand{\Qq}{\ensuremath{\mathcal{Q}}}
\newcommand{\Ss}{\ensuremath{\mathcal{S}}}
\newcommand{\app}{$\star$}

\newcommand{\prt}{\mathcal{N}}
\newcommand{\ola}{{\textsc{OLA}}\xspace}
\newcommand{\fas}{{\textsc{FAS}}\xspace}
\newcommand{\fastt}{{\textsc{FAST}}\xspace}
\newcommand{\ctw}{{\mathbf{ctw}}}

\newcommand{\defproblemu}[3]{
  \vspace{1mm}
\noindent\fbox{
  \begin{minipage}{0.95\textwidth}
  #1 \\
  {\bf{Input:}} #2  \\
  {\bf{Question:}} #3
  \end{minipage}
  }
  \vspace{1mm}
}
\newcommand{\defparproblem}[4]{
  \vspace{3mm}
\noindent\fbox{
  \begin{minipage}{0.96\textwidth}
  \begin{tabular*}{\textwidth}{@{\extracolsep{\fill}}lr} #1  & {\bf{Parameter:}} #3 \\ \end{tabular*}
  {\bf{Input:}} #2  \\
  {\bf{Question:}} #4
  \end{minipage}
  }
  \vspace{3mm}
}

\title{Subexponential parameterized algorithm for computing the~cutwidth of a semi-complete digraph}

\author{Fedor V. Fomin
\thanks{Department of Informatics, University of Bergen,  Bergen, Norway,  \{fomin,michal.pilipczuk\}@ii.uib.no. Supported by 
  the European Research Council (ERC) via grant Rigorous Theory of Preprocessing, reference 267959.} \and  Micha\l{} Pilipczuk~$^\dagger$
}

%


\maketitle

\begin{abstract}
Cutwidth of a digraph is a width measure introduced by Chudnovsky, Fradkin, and Seymour~\cite{ChudnovskyFS2011} in connection with development of a structural theory for tournaments, or more generally, for semi-complete digraphs. In this paper we provide an algorithm with running time $2^{O(\sqrt{k\log k})}\cdot n^{O(1)}$ that tests whether the cutwidth of a given $n$-vertex semi-complete digraph is at most $k$, improving upon the currently fastest algorithm of the second author~\cite{ja} that works in $2^{O(k)}\cdot n^2$ time. As a byproduct, we obtain a new algorithm for \fast in tournaments (\fastt) with running time $2^{c\sqrt{k}}\cdot n^{O(1)}$, where $c=\frac{2\pi}{\sqrt{3}\cdot \ln 2}\leq 5.24$, that is simpler than the algorithms of Feige~\cite{Feige09} and of Karpinski and~Schudy\cite{KarpinskiS10}, both also working in $2^{O(\sqrt{k})}\cdot n^{O(1)}$ time. Our techniques
can be applied also to other layout problems on semi-complete digraphs. We show that the \olaf problem, a close relative of \fast, can be solved in $2^{O(k^{1/3}\cdot\sqrt{\log k})}\cdot n^{O(1)}$ time, where $k$ is the target cost of the ordering.
\end{abstract}

\input{intro}

\input{preliminaries}

\input{partitions}

\input{algorithm}

\input{conclusions}

\bibliographystyle{siam}
\bibliography{cutwidth}

\end{document}

%% file: intro.tex
\section{Introduction}\label{sec:intro}

A directed graph is {\emph{simple} if it contains no multiple arcs or loops; it is moreover {\emph{semi-complete}} if for every two vertices $v,w$, at least one of the arcs $(v,w)$ or $(w,v)$ is present. An important subclass of semi-complete digraphs is the class of {\emph{tournaments}}, where we require that exactly one of these arcs is present. Tournaments are extensively studied both from combinatorial and computational point of view; see the book of Bang-Jensen  and Gutin \cite{BangG089_book} for an overview. 

One reason why the class of semi-complete digraphs is so interesting, is that for this class it is possible to construct a structural theory, resembling the theory of minors for undirected graphs. This theory has been developed recently by Chudnovsky, Fradkin, Kim, Scott, and Seymour~\cite{ChudnovskyFS2011,ChudnovskySS2011,ChudnovskyS11,fradkin-seymourEDP,fradkin-seymour,kim-seymour-minors}. In particular, two natural notions of digraph containment, namely immersion and minor orders, have been proven to well-quasi-order the set of semi-complete digraphs~\cite{ChudnovskyS11,kim-seymour-minors}. The developed structural theory has many algorithmic consequences, including fixed-parameter tractable algorithms for containment testing problems~\cite{ChudnovskyFS2011,my,ja}.

In both theories, for undirected graphs and semi-complete digraphs, width parameters play crucial roles. While in the theory of undirected graph minors the main parameter is treewidth, for semi-complete digraphs {\emph{cutwidth}} becomes one of the key notions. Given a semi-complete digraph $T$ and a vertex ordering $\sigma=(v_1,v_2,\ldots,v_n)$ of $V(T)$, the {\emph{width}} of $\sigma$ is $\max_{1\leq t\leq n-1} |E(\{v_{t+1},v_{t+2},\ldots,v_n\},\{v_{1},v_{2},\ldots,v_t\})|$, i.e., the maximum number of arcs that are directed from a suffix of the ordering to the complementary prefix. The {\emph{cutwidth}} of $T$, denoted $\ctw(T)$, is   the smallest possible width of an ordering of $V(T)$. It turns out that excluding a fixed digraph as an immersion implies an upper bound on cutwidth of a semi-complete digraph~\cite{ChudnovskyFS2011}. Hence, the claim that the immersion relation is a well-quasi-ordering of semi-complete digraphs can be easily reduced to the case of semi-complete digraphs of bounded cutwidth; there, a direct reasoning can be applied~\cite{ChudnovskyS11}.

From the computational point of view, Chudnovsky, Fradkin, and Seymour give an approximation algorithm that, given a semi-complete digraph $T$ on $n$ vertices and an integer $k$, in time $O(n^3)$ either outputs an ordering of width $O(k^2)$ or concludes that $\ctw(T)>k$~\cite{ChudnovskyFS2011}. It also follows from the work of Chudnovsky, Fradkin, and Seymour~\cite{ChudnovskyFS2011,ChudnovskySS2011} that the value of cutwidth can be computed exactly by a non-uniform fixed-parameter algorithm working in $f(k)\cdot n^3$ time: as the class of semi-complete digraphs of cutwidth at most $k$ is characterized by a finite set of forbidden immersions, we can approximate cutwidth and test existence of any of them using dynamic programming on the approximate ordering. These results were further improved by the second author~\cite{ja}: he gives an $O(OPT)$-approximation in $O(n^2)$ time by proving that any ordering of $V(T)$ according to outdegrees has width at most $O(\ctw(T)^2)$, and a fixed-parameter algorithm that in $2^{O(k)}\cdot n^2$ time finds an ordering of width at most $k$ or concludes that it is not possible.

\paragraph{Our results and techniques.} In this work we present an algorithm that, given a semi-complete digraph $T$ on $n$ vertices and an integer $k$, 
in $2^{O(\sqrt{k\log k})}\cdot n^{O(1)}$ time computes an ordering of width at most $k$ or concludes that $\ctw(T)>k$. In other words, we prove that the cutwidth of a semi-complete digraph can be computed in subexponential parameterized time.

The idea behind  our approach is inspired by the  recent work of a superset of the current authors on clustering problems~\cite{clustering}. The algorithm in~\cite{clustering} is based on a combinatorial result that in every YES instance of the problem the number of $k$-cuts, i.e., partitions of the vertex set into two subsets with at most $k$ edges crossing the partition, is bounded by a subexponential function of $k$. We apply a similar strategy to compute the cutwidth of a semi-complete digraph.  A {\emph{$k$-cut}} in a semi-complete digraph $T$ is a partition of its vertices into two sets $X$ and $Y$, such that only at most $k$ arcs are directed from $Y$ to $X$. Our algorithm is based on a new combinatorial lemma that the number of $k$-cuts in a semi-complete digraph of cutwidth at most $k$ is at most $2^{O(\sqrt{k\log k})}\cdot n$. The crucial ingredient of its proof is to relate $k$-cuts of a transitive tournament to partition numbers: a notion extensively studied in classical combinatorics and which   subexponential asymptotics is very well understood. Then we roughly do the following. It is possible to show that all $k$-cuts can be enumerated with polynomial time delay.  We enumerate all $k$-cuts and if   we exceed the combinatorial bound,  we are able to say that the cutwidth of the input digraph is more than $k$. Otherwise, we  have a bounded-size family of objects on which we can employ a dynamic programming routine. The running time of this step is up to a polynomial factor proportional to the number of $k$-cuts, and thus is subexponential. 


As a byproduct of the approach taken, we also obtain a new algorithm for \fast (\fas) in semi-complete digraphs, with running time $2^{c\sqrt{k}}\cdot n^{O(1)}$ for $c=\frac{2\pi}{\sqrt{3}\cdot \ln 2}\leq 5.24$. The \fas problem was  the first problem in tournaments shown to admit subexponential parameterized algorithms. The first algorithm with running time $2^{O(\sqrt{k}\log{k})}\cdot n^{O(1)}$ is due to Alon, Lokshtanov, and Saurabh~\cite{AlonLS09}. This has been further improved by Feige~\cite{Feige09} and by Karpinski and Schudy~\cite{KarpinskiS10}, who have independently shown two different algorithms with running time $2^{O(\sqrt{k})}\cdot n^{O(1)}$. The algorithm of Alon et al. introduced a new technique called~\emph{chromatic coding} that proved to be useful also in other problems in dense graphs \cite{ghosh2012faster}. The algorithms of Feige and of Karpinski and Schudy were based on the degree ordering approach, and the techniques developed there were more contrived to the problem.

In our approach, the $2^{O(\sqrt{k})}\cdot n^{O(1)}$ algorithm for \fas on semi-complete digraphs follows immediately from relating $k$-cuts of a transitive tournament to partition numbers, and an application of the general framework. It is also worth mentioning that the explicit constant in the exponent obtained using our approach is much smaller than the constants in the algorithms of Feige and of Karpinski and Schudy; however, optimizing these constants was not the purpose of these works. Similarly to the algorithm of Karpinski and Schudy, our algorithm works also in the weighted setting.

Lastly, we show that our approach can be also applied to other layout problems in semi-complete digraphs. For example, we consider a natural variant of the well-studied \olaf problem  \cite{ChinnCDG82,DiazPS02}, and we prove that one can compute in $2^{O(k^{1/3}\cdot\sqrt{\log k})}\cdot n^{O(1)}$ time an ordering of cost at most $k$, or conclude that it is impossible (see Section~\ref{sec:preliminaries} for precise definitions). Although such a low complexity may be explained by the fact that the optimal cost may be even cubic in the number of vertices, we find it interesting that results of this kind can be also obtained by making use of  our techniques.

\paragraph{Organization of the paper.} In Section~\ref{sec:preliminaries} we introduce basic notions and problem definitions. In Section~\ref{sec:partitions} we prove combinatorial lemmata concerning $k$-cuts of semi-complete digraphs. In Section~\ref{sec:algorithms} we apply the observations from the previous section to obtain the algorithmic results. Section~\ref{sec:conclusions} is devoted to concluding remarks.

%% file: preliminaries.tex
\section{Preliminaries}\label{sec:preliminaries}

We use standard graph notation. For a digraph $D$, we denote by $V(D)$ and $E(D)$ the vertex and edge sets of $D$, respectively. A digraph is {\emph{simple}} if it has no loops and no multiple arcs, i.e., for every pair of vertices $v,w$, the arc $(v,w)$ appears in $E(D)$ at most once. Note that we do not exclude existence of arcs $(v,w)$ and $(w,v)$ at the same time. All the digraphs considered in this paper will be simple

A digraph is {\emph{acyclic}} if it contains no cycle. It is known that a digraph is acyclic if and only if it admits a {\emph{topological ordering}} of vertices, i.e., an ordering $(v_1,v_2,\ldots,v_n)$ of $V(D)$ such that arcs are always directed from a vertex with a smaller index to a vertex with a larger index. 

A simple digraph $T$ is {\emph{semi-complete}} if for every pair $(v,w)$ of vertices {\bf{at least}} one of the arcs $(v,w)$, $(w,v)$ is present. A semi-complete digraph $T$ is moreover a {\emph{tournament}} if for every pair $(v,w)$ of vertices {\bf{exactly}} one of the arcs $(v,w)$, $(w,v)$ is present. A {\emph{transitive tournament}} with ordering $(v_1,v_2,\ldots,v_n)$ is a tournament $T$ defined on vertex set $\{v_1,v_2,\ldots,v_n\}$ where $(v_i,v_j)\in E(T)$ if and only if $i<j$.

We use Iverson notation: for a condition $\varphi$, $[\varphi]$ denotes value $1$ if $\varphi$ is true and $0$ otherwise. We also use $\exp(t)=e^t$.

\subsection{{\textsc{Feedback Arc Set}}}

\begin{definition}
Let $T$ be a digraph. A subset $F\subseteq E(T)$ is called a {\emph{feedback arc set}} if $T\setminus F$ is acyclic.
\end{definition}

The \fas problem in semi-complete digraphs is defined as follows.

\defparproblem{{\textsc{Feedback Arc Set}}}{A semi-complete digraph $T$, an integer $k$}{$k$}{Is there a feedback arc set of $T$ of size at most $k$?}

We have the following easy observation that enables us to view \fas as a graph layout problem.

\begin{lemma}\label{lem:fast-order}
Let $T$ be a digraph. Then $T$ admits a feedback arc set of size at most $k$ if and only if there exists an ordering $(v_1,v_2,\ldots,v_n)$ of $V(T)$ such that at most $k$ arcs of $E(T)$ are directed backward in this ordering, i.e., of form $(v_i,v_j)$ for $i>j$.
\end{lemma}
\begin{proof}
If $F$ is a feedback arc set in $T$ then the ordering can be obtained by taking any topological ordering of $T\setminus F$. On the other hand, given the ordering we may simply define $F$ to be the set of backward edges.
\end{proof}

\subsection{{\textsc{Cutwidth}}}

\begin{definition}
Let $T$ be a digraph. For an ordering $\sigma=(v_1,v_2,\ldots,v_n)$ of $V(T)$, the {\emph{width}} of $\sigma$ is $\max_{1\leq t\leq n-1} |E(\{v_{t+1},v_{t+2},\ldots,v_n\},\{v_{1},v_{2},\ldots,v_t\})|$. The {\emph{cutwidth}} of $T$, denoted $\ctw(T)$, is equal to the smallest possible width of an ordering of $V(T)$.
\end{definition}

The {\textsc{Cutwidth}} problem can be hence defined as follows:

\defparproblem{{\textsc{Cutwidth}}}{A semi-complete digraph $T$, an integer $k$}{$k$}{Is $\ctw(T)\leq k$?}

Note that the introduced notion reverses the ordering with respect to the standard literature on cutwidth~\cite{ChudnovskyFS2011,ja}. We chose to do so to be consistent within the paper, and compatible with the literature on \fas{} in tournaments. 

\subsection{{\textsc{Optimal Linear Arrangement}}}

\begin{definition}
Let $T$ be a digraph and $(v_1,v_2,\ldots,v_n)$ be an ordering of its vertices. Then the {\emph{cost}} of this ordering is defined as
$$\sum_{(v_i,v_j)\in E(T)} (i-j)\cdot [i>j],$$
that is, every arc directed backwards in the ordering contributes to the cost with the distance between the endpoints in the ordering.
\end{definition}

Whenever the ordering is clear from the context, we also refer to the contribution of a given arc to its cost as to the {\emph{length}} of this arc. By a simple reordering of the computation we obtain the following:

\begin{lemma}\label{lem:reordering}
For a digraph $T$ and ordering $(v_1,v_2,\ldots,v_n)$ of $V(T)$, the cost of this ordering is equal to:
$$\sum_{t=1}^{n-1} |E(\{v_{t+1},v_{t+2},\ldots,v_n\},\{v_1,v_2,\ldots,v_t\})|.$$
\end{lemma}
\begin{proof}
Observe that
\begin{eqnarray*}
\sum_{(v_i,v_j)\in E(T)} (i-j)\cdot [i>j] & = & \sum_{(v_i,v_j)\in E(T)}\ \sum_{t=1}^{n-1} [j\leq t<i] \\
& = & \sum_{t=1}^{n-1}\ \sum_{(v_i,v_j)\in E(T)} [j\leq t<i] \\ 
& = & \sum_{t=1}^{n-1} |E(\{v_{t+1},v_{t+2},\ldots,v_n\},\{v_1,v_2,\ldots,v_t\})|.
\end{eqnarray*}
\end{proof}

The problem \ola{} ({\textsc{Optimal Linear Arrangement}}) in semi-complete digraphs is defined as follows:

\defparproblem{{\textsc{Optimal Linear Arrangement}}}{A semi-complete digraph $T$, an integer $k$}{$k$}{Is there an ordering of $V(T)$ of cost at most $k$?}

%% file: partitions.tex
\section{$k$-cuts of semi-complete digraphs}\label{sec:partitions}

In this section we provide all the relevant observations on $k$-cuts of semi-complete digraphs. We start with the definitions, and then proceed to bounding the number of $k$-cuts when the given semi-complete digraph is close to a structured one.

\subsection{Definitions}

\begin{definition}
A $k${\emph{-cut}} of a digraph $T$ is a partition $(X,Y)$ of $V(T)$ with the following property: there are at most $k$ arcs $(u,v)\in E(T)$ such that $u\in Y$ and $v\in X$.
\end{definition}

The following lemma will be needed to apply the general framework.

\begin{lemma}\label{lem:enumeration}
$k$-cuts of a digraph $T$ can be enumerated with polynomial-time delay.
\end{lemma}
\begin{proof}
Let $\sigma=(v_1,v_2,\ldots,v_n)$ be an arbitrary ordering of vertices of $T$. We perform a classical branching strategy: we start with empty $X$ and $Y$, and consider the vertices in order $\sigma$, at each step branching into one of the two possibilities: vertex $v_i$ is to be incorporated into $X$ or into $Y$. However, after assigning each consecutive vertex we run a max-flow algorithm from $Y$ to $X$ to find the size of a minimum edge cut between $Y$ and $X$. If this size is more than $k$, we terminate the branch as we know that it cannot result in any solutions found. Otherwise we proceed. We output a partition after the last vertex, $v_n$, is assigned a side; note that the last max-flow check ensures that the output partition is actually a $k$-cut. Moreover, as during the algorithm we consider only branches that can produce at least one $k$-cut, the next partition will be always found within polynomial waiting time, proportional to the depth of the branching tree times the time needed for computations at each node of the branching tree.
\end{proof}

\subsection{$k$-cuts of a transitive tournament and partition numbers}

For a nonnegative integer $n$, a partition of $n$ is a multiset of positive integers whose sum is equal to $n$. The partition number $p(n)$ is equal to the number of different partitions of $n$. Partition numbers are studied extensively in analytic combinatorics, and there are sharp estimates on their value. In particular, we will use the following:

\begin{lemma}[\cite{erdos,hardy-ramanujan}]\label{lem:hardy-littlewood}
There exists a constant $A$ such that for every nonnegative $k$ it holds that $p(k)\leq \frac{A}{k+1}\cdot \exp(C\sqrt{k})$, where $C=\pi\sqrt{\frac{2}{3}}$.
\end{lemma}

We remark that the original proof of Hardy and Ramanujan~\cite{hardy-ramanujan} shows moreover that the optimal constant $A$ tends to $\frac{1}{4\sqrt{3}}$ as $k$ goes to infinity. From now on, we adopt constants $A,C$ given by Lemma~\ref{lem:hardy-littlewood} in the notation. We use Lemma~\ref{lem:hardy-littlewood} to obtain the following result, which is the core observation of this paper.

\begin{lemma}\label{lem:trans-partitions}
Let $T$ is a transitive tournament with $n$ vertices and $k$ be a nonnegative integer. Then $T$ has at most $A\cdot \exp(C\sqrt{k})\cdot (n+1)$ $k$-cuts, where $A,C$ are defined as in Lemma~\ref{lem:hardy-littlewood}.
\end{lemma}
\begin{proof}
We prove that for any number $a$, $0\leq a\leq n$, the number of $k$-cuts $(X,Y)$ such that $|X|=a$ and $|Y|=n-a$, is bounded by $A\cdot \exp(C\sqrt{k})$; summing through all the possible values of $a$ proves the claim.

We naturally identify the vertices of $T$ with numbers $1,2,\ldots,n$, such that arcs of $T$ are directed from smaller numbers to larger. Let us fix some $k$-cut $(X,Y)$ such that $|X|=a$ and $|Y|=n-a$. Let $x_1<x_2<\ldots<x_a$ be the vertices of $X$.

Let $m_i=x_{i+1}-x_i-1$ for $i=0,1,\ldots,a$; we use convention that $x_0=0$ and $x_{a+1}=n+1$. In other words, $m_i$ is the number of elements of $Y$ that are between two consecutive elements of $X$. Observe that every element of $Y$ between $x_i$ and $x_{i+1}$ is the tail of exactly $a-i$ arcs directed from $Y$ to $X$: the heads are $x_{i+1},x_{i+2},\ldots,x_a$. Hence, the total number of arcs directed from $Y$ to $X$ is equal to $k'=\sum_{i=0}^a m_i\cdot (a-i)=\sum_{i=0}^a m_{a-i}\cdot i\leq k$. 

We define a partition of $k'$ as follows: we take $m_{a-1}$ times number $1$, $m_{a-2}$ times number $2$, and so on, up to $m_0$ times number $a$.
Clearly, a $k$-cut of $T$ defines a partition of $k'$ in this manner. We now claim that knowing $a$ and the partition of $k'$, we can uniquely reconstruct the $k$-cut $(X,Y)$ of $T$, or conclude that this is impossible. Indeed, from   the partition we obtain all the numbers $m_0,m_1,\ldots,m_{a-1}$, while $m_a$ can be computed as $(n-a)-\sum_{i=0}^{a-1} m_i$. Hence, we know exactly how large must be the intervals between consecutive elements of $X$, and how far is the first and the last element of $X$ from the respective end of the ordering, which uniquely defines sets $X$ and $Y$. The only possibilities of failure during reconstruction are that (i) the numbers in the partition are larger than $a$, or (ii) computed $m_a$ turns out to be negative; in these cases, the partition does not correspond to any $k$-cut. Hence, we infer that the number of $k$-cuts of $T$ having $|X|=a$ and $|Y|=n-a$ is bounded by the sum of partition numbers of nonnegative integers smaller or equal to $k$, which by Lemma~\ref{lem:hardy-littlewood} is bounded by $(k+1)\cdot \frac{A}{k+1}\cdot \exp(C\sqrt{k})=A\cdot \exp(C\sqrt{k})$. 
\end{proof}

\subsection{$k$-cuts of semi-complete digraphs with a small \fas}

We have the following simple fact.

\begin{lemma}\label{lem:fast-partitions}
Assume that $T$ is a semi-complete digraph with a feedback arc set $F$ of size at most $k$. Let $T'$ be a transitive tournament on the same set of vertices, with vertices ordered as in any topological ordering of $T\setminus F$. Then every $k$-cut of $T$ is also a $2k$-cut of $T'$.
\end{lemma}
\begin{proof}
The claim follows directly from the observation that if $(X,Y)$ is a $k$-cut in $T$, then at most $k$ additional arcs directed from $Y$ to $X$ can appear after introducing arcs in $T'$ in place of deleted arcs from $F$.
\end{proof}

From Lemmata~\ref{lem:trans-partitions} and~\ref{lem:fast-partitions} we obtain the following corollary.

\begin{corollary}\label{cor:fast-partitions}
Every semi-complete digraph with $n$ vertices and with a feedback arc set of size at most $k$, has at most $A\cdot \exp(C\sqrt{2k})\cdot (n+1)$ $k$-cuts.
\end{corollary}

\subsection{$k$-cuts of semi-complete digraphs of small cutwidth}

To bound the number of $k$-cuts of semi-complete digraphs of small cutwidth, we need the following auxiliary combinatorial result.

\begin{lemma}\label{lem:bad}
Let $(X,Y)$ be a partition of $\{1,2,\ldots,n\}$ into two sets. We say that a pair $(a,b)$ is {\emph{bad}} if $a<b$, $a\in Y$ and $b\in X$. Assume that for every integer $t$ there are at most $k$ bad pairs $(a,b)$ such that $a\leq t<b$. Then the total number of bad pairs is at most $k(1+\ln k)$.
\end{lemma}
\begin{proof}
Let $y_1<y_2<\ldots<y_p$ be the elements of $Y$. Let $m_i$ be equal to the total number of elements of $X$ that are greater than $y_i$. Note that $m_i$ is exactly equal to the number of bad pairs whose first element is equal to $y_i$, hence the total number of bad pairs is equal to $\sum_{i=1}^p m_i$. Clearly, sequence $(m_i)$ is non-decreasing, so let $p'$ be the last index for which $m_{p'}>0$. We then have that the total number of bad pairs is equal to $\sum_{i=1}^{p'} m_i$. Moreover, observe that $p'\leq k$, as otherwise there would be more than $k$ bad pairs $(a,b)$ for which $a\leq y_{p'}<b$: for $a$ we can take any $y_i$ for $i\leq p'$ and for $b$ we can take any element of $X$ larger than $y_{p'}$.

We claim that $m_i\leq k/i$ for every $1\leq i\leq p'$. Indeed, observe that there are exactly $i\cdot m_i$ bad pairs $(a,b)$ for $a\leq y_i$ and $b>y_i$: $a$ can be chosen among $i$ distinct integers $y_1,y_2,\ldots,y_i$, while $b$ can be chosen among $m_i$ elements of $X$ larger than $y_i$. By the assumption we infer that $i\cdot m_i\leq k$, so $m_i\leq k/i$. Concluding, we have that the total number of bad pairs is bounded by $\sum_{i=1}^{p'} m_i\leq \sum_{i=1}^{p'} k/i=k\cdot H(p')\leq k\cdot H(k)\leq k(1+\ln{k})$, where $H(k)=\sum_{i=1}^{k} 1/i$ is the harmonic function.
\end{proof}

The following claim applies Lemma~\ref{lem:bad} to the setting of semi-complete digraphs.

\begin{lemma}\label{lem:cutwidth-partitions}
Assume that $T$ is a semi-complete digraph with $n$ vertices that admits an ordering of vertices $(v_1,v_2,\ldots,v_n)$ of width at most $k$. Let $T'$ be a transitive tournament on the same set of vertices, where $(v_i,v_j)\in E(T')$ if and only if $i<j$. Then every $k$-cut of $T$ is a $2k(1+\ln 2k)$-cut of $T'$.
\end{lemma}
\begin{proof}
Without loss of generality we assume that $T$ is in fact a tournament, as deleting any of two opposite arcs connecting two vertices can only make the set of $k$-cuts of $T$ larger, and does not increase the width of the ordering.

Identify vertices $v_1,v_2,\ldots,v_n$ with numbers $1,2,\ldots,n$. Let $(X,Y)$ be a $k$-cut of $T$. Note that arcs of $T'$ directed from $Y$ to $X$ correspond to bad pairs in the sense of Lemma~\ref{lem:bad}. Therefore, by Lemma~\ref{lem:bad} is suffices to prove that for every integer $t$, the number of arcs $(a,b)\in E(T')$ such that $a\leq t<b$, $a\in Y$, and $b\in X$, is bounded by $2k$. We know that the number of such arcs in $T$ is at most $k$, as there are at most $k$ arcs directed from $Y$ to $X$ in $T$ in total. Moreover, as the considered ordering of $T$ has cutwidth at most $k$, at most $k$ arcs between vertices from $\{1,2,\ldots,t\}$ and $\{t+1,\ldots,n\}$ can be directed in different directions in $T$ and in $T'$. We infer that the number of arcs $(a,b)\in E(T')$ such that $a\leq t<b$, $a\in Y$, and $b\in X$, is bounded by $2k$, and so the lemma follows.
\end{proof}

From Lemmata~\ref{lem:trans-partitions} and~\ref{lem:cutwidth-partitions} we obtain the following corollary.

\begin{corollary}\label{cor:cutwidth-partitions}Every 
  tournament with $n$ vertices and of  cutwidth at most $k$, has at most $A\cdot \exp(2C\sqrt{k(1+\ln 2k)})\cdot (n+1)$ $k$-cuts. 
%
\end{corollary}

\subsection{$k$-cuts of semi-complete digraphs with an ordering of small cost}

We firstly show the following lemma that proves that semi-complete digraphs with an ordering of small cost have even smaller cutwidth.

\begin{lemma}\label{lem:ola-small-cutwidth}
Let $T$ be a semi-complete digraph on $n$ vertices that admits an ordering $(v_1,v_2,\ldots,v_n)$ of cost at most $k$. Then the width of this ordering is at most $(4k)^{2/3}$.
\end{lemma}
\begin{proof} We claim that for every 
  integer $t\geq 0$,   the number of arcs in $T$ directed from the set $\{v_{t+1},\ldots,v_n\}$ to $\{v_1,\ldots,v_t\}$ is at most  $(4k)^{2/3}$. Let $\ell$ be the number of such arcs; without loss of generality assume that $\ell>0$. Observe that at most one of these arcs may have length $1$, at most $2$ may have length $2$, etc., up to at most $\lfloor \sqrt{\ell}\rfloor-1$ may have length $\lfloor \sqrt{\ell}\rfloor-1$. It follows that at most $\sum_{i=1}^{\lfloor \sqrt{\ell}\rfloor-1} i\ \leq \ell/2$ of these arcs may have length smaller than $\lfloor\sqrt{\ell}\rfloor$. Hence, at least $\ell/2$ of the considered arcs have length at least $\lfloor\sqrt{\ell}\rfloor$, so the total sum of lengths of arcs is at least $\frac{\ell\cdot \lfloor\sqrt{\ell}\rfloor}{2}\geq \frac{\ell^{3/2}}{4}$. We infer that $k\geq \frac{\ell^{3/2}}{4}$, which means that $\ell \leq (4k)^{2/3}$.
\end{proof}

Lemma~\ref{lem:ola-small-cutwidth} ensures that only $(4k)^{2/3}$-cuts are interesting from the point of view of dynamic programming. Moreover, from Lemma~\ref{lem:ola-small-cutwidth} and Corollary~\ref{cor:cutwidth-partitions} we can derive the following statement that bounds the number of states of the dynamic program.

\begin{corollary}\label{cor:ola-partitions}
If $T$ is a semi-complete digraph with $n$ vertices that admits an ordering of cost at most $k$, then the number of $(4k)^{2/3}$-cuts of $T$ is bounded by $A\cdot \exp(2C\cdot (4k)^{1/3}\cdot \sqrt{1+\ln (2\cdot(4k)^{2/3})})\cdot (n+1)$.
\end{corollary}

%% file: algorithm.tex
\section{The algorithms}\label{sec:algorithms}

For a given semi-complete digraph $T$, let $\prt(T,k)$ denote the family of $k$-cuts of $T$. We firstly show how using our approach one can find a simple algorithm for \fast.

\begin{theorem}\label{thm:alg-fast}
There exists an algorithm that, given a semi-complete digraph $T$ on $n$ vertices and an integer $k$, in time $\exp(C\sqrt{2k})\cdot n^{O(1)}$ either finds a feedback arc set of $T$ of size at most $k$ or correctly concludes that this is impossible, where $C=\pi\sqrt{\frac{2}{3}}$ .
\end{theorem}
\begin{proof}
Using Lemma~\ref{lem:enumeration}, we enumerate all the $k$-cuts of $T$. If we exceed the bound of $A\cdot \exp(C\sqrt{2k})\cdot (n+1)$ during enumeration, by Corollary~\ref{cor:fast-partitions} we may safely terminate the computation providing a negative answer; note that this happens after using at most $\exp(C\sqrt{2k})\cdot n^{O(1)}$ time, as the cuts are output with polynomial time delay. Hence, from now on we assume that we have the set $\prt:=\prt(T,k)$ and we know that $|\prt|\leq A\cdot \exp(C\sqrt{2k})\cdot (n+1)$.

We now describe a dynamic programming procedure that computes the size of optimal feedback arc set basing on the set $\prt$; the dynamic program is based on the approach presented in~\cite{ja}. We define an auxiliary weighted digraph $D$ with vertex set $\prt$. Intuitively, a vertex from $\prt$ corresponds to a partition into prefix and suffix of the ordering.

Formally, we define arcs of $D$ as follows. We say that cut $(X_2,Y_2)$ {\emph{extends}} cut $(X_1,Y_1)$ if there is one vertex $v\in Y_1$ such that $X_2=X_1\cup \{v\}$ and, hence, $Y_2=Y_1\setminus \{v\}$. We put an arc in $D$ from cut $(X_1,Y_1)$ to cut $(X_2,Y_2)$ if $(X_2,Y_2)$ extends $(X_1,Y_1)$; the weight of this arc is equal to $|E(\{v\},X_1)|$, that is, the number of arcs that cease to be be directed from the right side to the left side of the partition when moving $v$ between these parts. Note that thus each vertex of $D$ has at most $n$ outneighbours, so $|E(D)|$ is bounded by $O(|\prt|\cdot n)$. Moreover, the whole graph $D$ can be constructed in $|\prt|\cdot n^{O(1)}$ time by considering all the vertices of $D$ and examining each of at most $n$ candidates for outneighbours in polynomial time.

Observe that a path from vertex $(\emptyset,V(T))$ to a vertex $(V(T),\emptyset)$ of total weight $\ell$ defines an ordering of vertices of $T$ that has exactly $\ell$ backward arcs --- each of these edges was taken into account while moving its tail from the right side of the partition to the left side. On the other hand, every ordering of vertices of $T$ that has exactly $\ell\leq k$ backward arcs defines a path from $(\emptyset,V(T))$ to $(V(T),\emptyset)$ in $D$ of total weight $\ell$; note that all partitions into prefix and suffix in this ordering are $k$-cuts, so they constitute legal vertices in $D$. Hence, we need to check whether vertex $(V(T),\emptyset)$ can be reached from $(\emptyset,V(T))$ by a path of total length at most $k$. This, however, can be done in time $O((|V(D)|+|E(D)|)\log |V(D)|)=O(\exp(C\sqrt{2k})\cdot n^{O(1)})$ using Dijkstra's algorithm. The feedback arc set of size at most $k$ can be easily retrieved from the constructed path in polynomial time.
\end{proof}

We remark that it is straightforward to adapt the algorithm of Theorem~\ref{thm:alg-fast} to the weighted case, where all the arcs are assigned a real weight larger or equal to $1$ and we parametrize by the target total weight of the solution. As the minimum weight is at least $1$, we may still consider only $k$-cuts of the digraph where the weights are forgotten. On this set we employ a modified dynamic programming routine, where the weights of arcs in digraph $D$ are not simply the number of arcs in $E(\{v\},X_1)$, but their total weight. We omit the details here.

We now proceed to the main result of this paper, i.e., the subexponential algorithm for cutwidth of a semi-complete digraph. The following theorem essentially follows from Lemma~\ref{lem:enumeration}, Corollary~\ref{cor:cutwidth-partitions} and the dynamic programming algorithm from~\cite{ja}; we include the proof for the sake of completeness.

\begin{theorem}\label{thm:alg-cutwidth}
There exists an algorithm that, given a semi-complete digraph $T$ on $n$ vertices and an integer $k$, in time $2^{O(\sqrt{k\log k})}\cdot n^{O(1)}$ either computes a vertex ordering of width at most $k$ or correctly concludes that this is impossible.
\end{theorem}
\begin{proof}
Using Lemma~\ref{lem:enumeration} we enumerate all the $k$-cuts of $T$. If we exceed the bound of $A\cdot \exp(2C\sqrt{k(1+\ln 2k)})\cdot (n+1)=2^{O(\sqrt{k\log k})}\cdot n$ during enumeration, by Corollary~\ref{cor:cutwidth-partitions} we may safely terminate the computation providing a negative answer; note that this happens after using at most $2^{O(\sqrt{k\log k})}\cdot n^{O(1)}$ time, as the cuts are output with polynomial time delay. Hence, from now on we assume that we have the set $\prt:=\prt(T,k)$ and we know that $|\prt|=2^{O(\sqrt{k\log k})}\cdot n$.

We now recall the dynamic programming procedure from~\cite{ja} that basing on the set $\prt$ computes an ordering of width at most $k$ or correctly concludes that it is impossible. 

We proceed very similarly to the proof of Theorem~\ref{thm:alg-fast}. Define an auxiliary digraph $D$ on the vertex set $\prt$, where we put an arc from cut $(X_1,Y_1)$ to cut $(X_2,Y_2)$ if and only if $(X_2,Y_2)$ extends $(X_1,Y_1)$. Clearly, paths in $D$ from $(\emptyset,V(T))$ to $(V(T),\emptyset)$ correspond to orderings of $V(T)$ of cutwidth at most $k$. Therefore, it suffices to construct digraph $D$ and run a depth-first search from the vertex $(\emptyset,V(T))$. Note that $D$ has at most $|\prt|\leq 2^{O(\sqrt{k\log k})}\cdot n$ vertices, and every vertex has at most $n$ outneighbours; hence $|E(D)|\leq 2^{O(\sqrt{k\log k})}\cdot n^2$. Therefore, we can construct $D$ in $2^{O(\sqrt{k\log k})}\cdot n^{O(1)}$ time by running through all the vertices and examining every candidate for an outneighbour in polynomial time. Then we can check whether there exists a path from from $(\emptyset,V(T))$ to $(V(T),\emptyset)$ using depth-first search; the corresponding ordering may be retrieved from this path in polynomial time.
\end{proof}

Finally, we present how the framework can be applied to the \ola problem.

\begin{theorem}\label{thm:alg-ola}
There exists an algorithm that, given a semi-complete digraph $T$ on $n$ vertices and an integer $k$, in time $2^{O(k^{1/3}\sqrt{\log k})}\cdot n^{O(1)}$ either computes a vertex ordering of cost at most $k$, or correctly concludes that it is not possible.
\end{theorem}
\begin{proof}
Using Lemma~\ref{lem:enumeration} we enumerate the $(4k)^{2/3}$-cuts of $T$. If we exceed the bound of $A\cdot \exp(2C\cdot (4k)^{1/3}\cdot \sqrt{1+\ln (2\cdot(4k)^{2/3})})\cdot (n+1)=2^{O(k^{1/3}\sqrt{\log k})}\cdot n$ during enumeration, by Corollary~\ref{cor:ola-partitions} we may safely terminate the computation providing a negative answer; note that this happens after using at most $2^{O(k^{1/3}\sqrt{\log k})}\cdot n^{O(1)}$ time, as the cuts are output with polynomial time delay. Hence, from now on we assume that we have the set $\prt:=\prt(T,(4k)^{2/3})$ and we know that $|\prt|=2^{O(k^{1/3}\sqrt{\log k})}\cdot n$.

We proceed very similarly to the proof of Theorem~\ref{thm:alg-fast}. Define an auxiliary digraph $D$ on the vertex set $\prt$, where we put an arc from $(X_1,Y_1)$ to $(X_2,Y_2)$ if and only if $(X_2,Y_2)$ extends $(X_1,Y_1)$; the weight of this arc is equal to $|E(Y_1,X_1)|$. As in the proof of Theorem~\ref{thm:alg-fast}, paths from $(\emptyset,V(T))$ to $(V(T),\emptyset)$ of total weight $\ell\leq k$ correspond one-to-one to orderings with cost~$\ell$: the weight accumulated along the path computes correctly the cost of the ordering due to Lemma~\ref{lem:reordering}. Note that Lemma~\ref{lem:ola-small-cutwidth} ensures that in an ordering of cost at most $k$, the only feasible partitions into prefix and suffix of the ordering are in $\prt$, so they constitute legal vertices in $D$. 

Similarly as in the proof of Theorems~\ref{thm:alg-fast} and~\ref{thm:alg-cutwidth}, $D$ has $2^{O(k^{1/3}\sqrt{\log k})}\cdot n^{O(1)}$ vertices and arcs, and can be constructed in 
$2^{O(k^{1/3}\sqrt{\log k})}\cdot n^{O(1)}$  time. Hence, we may apply Dijkstra's algorithm to check whether vertex $(V(T),\emptyset)$ is reachable from $(\emptyset,V(T))$ via a path of length at most $k$. The corresponding ordering may be retrieved from this path in polynomial time.
\end{proof}

Similarly to Theorem~\ref{thm:alg-fast}, it is also straightforward to adapt the algorithm of Theorem~\ref{thm:alg-ola} to the natural weighted variant of the problem, where each arc is assigned a real weight larger or equal to $1$, each arc directed backward in the ordering contributes to the cost with its weight multiplied by the length of the arc, and we parametrize by the total target cost.

%% file: conclusions.tex
\section{Conclusions}\label{sec:conclusions}

In this paper we showed that a number of vertex ordering problems on tournaments, and more generally, on semi-complete digraphs, admit subexponential parameterized algorithms. 

We believe that our approach provides a deep insight into the structure of problems on semi-complete digraphs solvable in subexponential parameterized time: in instances with a positive answer, the space of naturally relevant objects, namely $k$-cuts, is of subexponential size. We hope that this kind of algorithm design strategy may be applied to other problems as well. 

Clearly, it is possible to pipeline the presented algorithm for \fas on semi-complete digraphs with a simple kernelization algorithm, which can be found, e.g., in~\cite{AlonLS09}, to separate the polynomial dependency on $n$ from the subexponential dependency on $k$ in the running time. This can be done also for \ola, as this problem admits a simple linear kernel; we omit the details here. 

However, we believe that a more important challenge is to investigate whether the $\sqrt{\log k}$ factor in the exponent of the running times of the algorithms of Theorems~\ref{thm:alg-cutwidth} and~\ref{thm:alg-ola} is necessary. At this moment, appearance of this factor is a result of pipelining Lemma~\ref{lem:trans-partitions} with Lemma~\ref{lem:bad} in the proof of Lemma~\ref{lem:cutwidth-partitions}. A closer examination of the proofs of Lemmata~\ref{lem:trans-partitions} and~\ref{lem:bad} shows that bounds given by them are essentially optimal on their own; yet, it is not clear whether the bound given by pipelining them is optimal as well. Hence, we would like to pose the following open problem: is the number of $k$-cuts of a semi-complete digraph on $n$ vertices and of cutwidth at most $k$ bounded by $2^{O(\sqrt{k})}\cdot n^{O(1)}$? If the answer to this combinatorial question is positive, then the $\sqrt{\log k}$ factor could be removed.

%% file: cutwidth.bbl
\begin{thebibliography}{10}

\bibitem{AlonLS09}
{\sc N.~Alon, D.~Lokshtanov, and S.~Saurabh}, {\em Fast {FAST}}, in ICALP,
  vol.~5555 of LNCS, Springer, 2009, pp.~49--58.

\bibitem{BangG089_book}
{\sc J.~Bang-Jensen and G.~Gutin}, {\em Digraphs}, Springer Monographs in
  Mathematics, Springer-Verlag London Ltd., London, second~ed., 2009.
\newblock Theory, algorithms and applications.

\bibitem{ChinnCDG82}
{\sc P.~Z. Chinn, J.~Chv\'{a}talov\'{a}, A.~K. Dewdney, and N.~E. Gibbs}, {\em
  The bandwidth problem for graphs and matrices --- a survey}, J. Graph Theory,
  6 (1982), pp.~223--254.

\bibitem{ChudnovskyFS2011}
{\sc M.~Chudnovsky, A.~Fradkin, and P.~Seymour}, {\em Tournament immersion and
  cutwidth}, J. Comb. Theory Ser. B, 102 (2012), pp.~93--101.

\bibitem{ChudnovskySS2011}
{\sc M.~Chudnovsky, A.~Scott, and P.~Seymour}, {\em Vertex disjoint paths in
  tournaments}, 2011.
\newblock Manuscript.

\bibitem{ChudnovskyS11}
{\sc M.~Chudnovsky and P.~D. Seymour}, {\em A well-quasi-order for
  tournaments}, J. Comb. Theory, Ser. B, 101 (2011), pp.~47--53.

\bibitem{DiazPS02}
{\sc J.~D\'{\i}az, J.~Petit, and M.~J. Serna}, {\em A survey of graph layout
  problems}, ACM Comput. Surv., 34 (2002), pp.~313--356.

\bibitem{erdos}
{\sc P.~Erd\H{o}s}, {\em On an elementary proof of some asymptotic formulas in
  the theory of partitions}, Annals of Mathematics (2), 43 (1942),
  pp.~437--450.

\bibitem{Feige09}
{\sc U.~Feige}, {\em Faster {FAST} ({F}eedback {A}rc {S}et in {T}ournaments)},
  CoRR, abs/0911.5094 (2009).

\bibitem{clustering}
{\sc F.~V. Fomin, S.~Kratsch, M.~Pilipczuk, M.~Pilipczuk, and Y.~Villanger},
  {\em Subexponential fixed-parameter tractability of cluster editing}, CoRR,
  abs/1112.4419 (2011).
\newblock To appear in the proceedings of STACS 2013.

\bibitem{my}
{\sc F.~V. Fomin and M.~Pilipczuk}, {\em Jungles, bundles, and fixed parameter
  tractability}, in Proceedings of the 24th ACM-SIAM Symposium on Discrete
  Algorithms (SODA), SIAM, 2012, pp.~396--413.

\bibitem{fradkin-seymourEDP}
{\sc A.~Fradkin and P.~Seymour}, {\em Edge-disjoint paths in digraphs with
  bounded independence number}, 2010.
\newblock Manuscript.

\bibitem{fradkin-seymour}
\leavevmode\vrule height 2pt depth -1.6pt width 23pt, {\em Tournament pathwidth
  and topological containment}, 2011.
\newblock Manuscript.

\bibitem{ghosh2012faster}
{\sc E.~Ghosh, S.~Kolay, M.~Kumar, P.~Misra, F.~Panolan, A.~Rai, and
  M.~Ramanujan}, {\em Faster parameterized algorithms for deletion to split
  graphs}, in Proceedings of the 13th Scandinavian Symposium and Workshops on
  Algorithm Theory (SWAT), vol.~7357 of Lecture Notes in Computer Science,
  Springer, 2012, pp.~107--118.

\bibitem{hardy-ramanujan}
{\sc G.~H. Hardy and S.~Ramanujan}, {\em Asymptotic formulae in combinatory
  analysis}, Proceedings of the London Mathematical Society, s2-17 (1918),
  pp.~75--115.

\bibitem{KarpinskiS10}
{\sc M.~Karpinski and W.~Schudy}, {\em Faster algorithms for feedback arc set
  tournament, {K}emeny rank aggregation and betweenness tournament}, in
  Proceedings of the 21st International Symposium on Algorithms and Computation
  (ISAAC), vol.~6506 of Lecture Notes in Computer Science, Springer, 2010,
  pp.~3--14.

\bibitem{kim-seymour-minors}
{\sc I.~Kim and P.~Seymour}, {\em Tournament minors}, CoRR, abs/1206.3135
  (2012).

\bibitem{ja}
{\sc M.~Pilipczuk}, {\em Computing cutwidth and pathwidth of semi-complete
  digraphs via degree orderings}, CoRR, abs/1210.5363 (2012).
\newblock To appear in the proceedings of STACS 2013.

\end{thebibliography}
